\documentclass[a4paper,UKenglish,cleveref, autoref]{lipics-v2019}
\renewcommand{\O}[1]{\ensuremath{{O(#1)}}} %\mathcal{O}
\newcommand{\etal}{\textit{e{}t~a{}l.}\xspace}

\usepackage{dsfont}

\renewcommand{\Pr}[1]{\ensuremath{\mathbf{Pr}\left[#1\right]}}
\newcommand{\Ex}[1]{\ensuremath{\mathbb{E}\left[#1\right]}}
\newcommand{\Exp}[2]{\ensuremath{\mathbb{E}_{#1}\left[#2\right]}}

\newcommand{\norm}[1]{\ensuremath{\left\| #1 \right\|}}
\newcommand{\REAL}{\ensuremath{\mathbb{R}}}

\newcommand{\removed}[1]{}
\newcommand{\eps}{\varepsilon}
\newcommand{\br}[1]{\left\{#1\right\}}                            % {#}
\newcommand{\medgen}{set median }

\usepackage[ruled,vlined,linesnumbered]{algorithm2e}
\SetFuncSty{textsc}
\SetDataSty{textit}
\SetCommentSty{textrm}
\SetKwFunction{Extract}{ExtractFirstDigit}
\SetKwData{Forward}{forward}
\SetKw{True}{true}
\SetEndCharOfAlgoLine{}

\usepackage{microtype}%if unwanted, comment out or use option "draft"

\title{Probabilistic Smallest Enclosing Ball in High Dimensions via Subgradient Sampling} 

\titlerunning{Probabilistic Smallest Enclosing Ball in High Dimensions}%optional, please use if title is longer than one line

\author{Amer Krivo\v{s}ija}{Department of Computer Science, TU Dortmund, Germany }{amer.krivosija@tu-dortmund.de}{}{}%TODO mandatory, please use full name; only 1 author per \author macro; first two parameters are mandatory, other parameters can be empty. Please provide at least the name of the affiliation and the country. The full address is optional

\author{Alexander Munteanu}{Department of Computer Science, TU Dortmund, Germany }{alexander.munteanu@tu-dortmund.de}{}{}

\authorrunning{A. Krivo\v{s}ija and A. Munteanu}%TODO mandatory. First: Use abbreviated first/middle names. Second (only in severe cases): Use first author plus 'et al.'

\Copyright{Amer Krivo\v{s}ija and Alexander Munteanu}%TODO mandatory, please use full first names. LIPIcs license is "CC-BY";  http://creativecommons.org/licenses/by/3.0/


%TODO mandatory; please add comma-separated list of keywords

\relatedversion{}%optional, e.g. full version hosted on arXiv, HAL, or other respository/website
\funding{This work was supported by the German Science Foundation (DFG) Collaborative Research Center SFB 876 "Providing Information by Resource-Constrained Analysis", projects A2 and C4.}

\EventEditors{Gill Barequet and Yusu Wang}
\EventNoEds{2}
\EventLongTitle{35th International Symposium on Computational Geometry (SoCG 2019)}
\EventShortTitle{SoCG 2019}
\EventAcronym{SoCG}
\EventYear{2019}
\EventDate{June 18--21, 2019}
\EventLocation{Portland, United~States}
\EventLogo{socg-logo}
\SeriesVolume{129}
\ArticleNo{0} % <- Your article number goes here
\nolinenumbers %uncomment to disable line numbering
\hideLIPIcs  %uncomment to remove references to LIPIcs series (logo, DOI, ...), e.g. when preparing a pre-final version to be uploaded to arXiv or another public repository
\begin{document}

\maketitle

\begin{abstract}
We study a variant of the median problem for a collection of point sets in high dimensions. This generalizes the geometric median as well as the (probabilistic) smallest enclosing ball (pSEB) problems. Our main objective and motivation is to improve the previously best algorithm for the pSEB problem by reducing its exponential dependence on the dimension to linear. This is achieved via a novel combination of sampling techniques for clustering problems in metric spaces with the framework of stochastic subgradient descent. As a result, the algorithm becomes applicable to shape fitting problems in Hilbert spaces of unbounded dimension via kernel functions. We present an exemplary application by extending the support vector data description (SVDD) shape fitting method to the probabilistic case. This is done by simulating the pSEB algorithm implicitly in the feature space induced by the kernel function.
\end{abstract}

\section{Introduction}
%%%% Alex Text
%\clearpage
The (probabilistic) smallest enclosing ball (pSEB) problem {in $\REAL^d$} is to find a center that minimizes the (expected) maximum distance to the input points (see Definition~\ref{def:probseb}). It occurs often as a building block for complex data analysis and machine learning tasks like estimating the support of high dimensional distributions, outlier detection, novelty detection, classification and robot gathering \cite{CieliebakFPS12,StolpeBDM13,TaxD04,TsangKC05}. It is thus very important to develop highly efficient approximation algorithms for the base problem. This involves reducing the number of points but also keeping the dependence on the dimension as low as possible. Both objectives will be studied in this paper. We will focus on a small dependence on the dimension. This is motivated as follows.

Kernel methods are a common technique in machine learning. These methods implicitly project the $d$-dimensional input data into much larger dimension $D$ where simple linear classifiers or spherical data fitting methods can be applied to obtain a non-linear separation or non-convex shapes in the original $d$-dimensional space. The efficiency of kernel methods is usually not harmed. Despite the large dimension $D\gg d$, most important kernels, and thus inner products and distances in the $D$-dimensional space, can be evaluated in $O(d)$ time \cite{RasmussenW06}.

In some cases, however, a proper approximation relying on sampling and discretizing the ambient solution space may require a polynomial or even exponential dependence on $D$. The algorithm of Munteanu \etal \cite{MunteanuSF14} is the only fully polynomial time approximation scheme (FPTAS) and the fastest algorithm to date for the pSEB problem in fixed dimension. However, it suffers from the stated problems. In particular, the number of realizations sampled by their algorithm had a linear dependence on $D$ stemming from a ball-cover decomposition of the solution space. The actual algorithm made a brute force evaluation (on the sample) of all centers in a grid of exponential size in $D$. This is prohibitive in the setting of kernel methods since the implicit feature space may have infinite dimension. Even if it is possible to exploit the up to $n$-dimensional subspace spanned by $n$ points in infinite dimensions, we would still have $D = n \gg d$ leading to exponential time algorithms.

To make the probabilistic smallest enclosing ball algorithm viable in the context of kernel methods and generally in high dimensions, it is highly desirable to reduce the dependence on the dimension to a small polynomial occurring only in evaluations of inner products and distances between two (low dimensional) vectors.

%\noindent
\subsection{Related work} 
\begin{description}
\item[Probabilistic smallest enclosing ball]
The study of \emph{probabilistic clustering problems} was initiated by Cormode and McGregor \cite{CormodeM08}. They developed approximation algorithms for the probabilistic settings of $k$-means, $k$-median as well as $k$-center clustering. For the metric $1$-center clustering problem their results are bi-criteria $O(1)$-approximation algorithms with a blow-up on the number of centers. Guha and Munagala \cite{GuhaM09} improved the previous work by giving $O(1)$-approximations that preserve the number of centers. Munteanu \etal \cite{MunteanuSF14} gave the first fully polynomial time $(1+\eps)$-approximation scheme (FPTAS) for the probabilistic Euclidean $1$-center, i.e., the probabilistic smallest enclosing ball problem, in fixed dimensions. The algorithm runs in linear time for sampling a constant number of realizations. Solving the subsampled problem takes only constant time, though exponential in the dimension. This yields a total running time of roughly $O(nd/\eps^{O(1)} + 1/\eps^{O(d)})$. Based on $\eps$-kernels for probabilistic data \cite{HuangLPW16}, Huang and Li \cite{HuangL17} generalized the $(1+\eps)$-approximation of \cite{MunteanuSF14} to a polynomial time approximation scheme (PTAS) for Euclidean $k$-center in $\REAL^d$ for fixed constants $k$ and $d$. We note that the running time of the algorithm in \cite{HuangL17} grows as a double exponential function of the dimension and it is unclear how to reduce this. We thus base our work on the FPTAS of \cite{MunteanuSF14} and reduce its exponential dependence to linear.

\item[Sampling techniques for $1$-median]
The work of \cite{MunteanuSF14} showed a reduction from the probabilistic smallest enclosing ball to $1$-median problems on \emph{(near)-metric} spaces. We thus review relevant results on sampling techniques for metric $1$-median.
B\u{a}doiu \etal \cite{Badoiu02} showed that with constant probability, the span of a uniform sample of a constant number of input points contains a $(1+\eps)$-approximation for the $1$-median. This was used to construct a set of candidate solutions of size $O(2^{1/\eps^{O(1)}}\log n)$. A more refined estimation procedure based again on a small uniform sample was used by Kumar \etal \cite{KumarSS10} to reduce this to $O(2^{1/\eps^{O(1)}})$. Indyk and Thorup~\cite{IndykDiss, Thorup04} showed that a uniform sample of size $O(\log n/\eps^2)$ is sufficient to approximate the discrete metric $1$-median on $n$ points within a factor of $(1+\eps)$. Ackermann \etal \cite{AckermannBS10} showed how this argument can be adapted to doubling spaces which include the continuous Euclidean space. We adapt these ideas to find a $(1+\eps)$-approximation to the best center in our setting.

\item[Stochastic subgradient descent]
One quite popular and often only implicitly used technique in the coreset literature is derived from convex optimization, see \cite{MunteanuS18}.
One of the first results of that kind is given in the uniform sampling algorithm of B\u{a}doiu \etal \cite{Badoiu02} for $1$-median stated above. In each iteration a single point is sampled uniformly. With high probability moving the current center towards that point for a carefully chosen step size improves the solution. Each step can be seen as taking a descent towards a uniformly random direction from the subgradient which equals roughly the sum of directions to all points. Another example is the coreset construction for the smallest enclosing ball problem by B\u{a}doiu and Clarkson \cite{BadoiuC03}, where the next point included in the coreset is the one maximizing the distance to the current best center. The direction taken towards that point is again related to the subgradient of the objective function at the current center position. The authors also gave a more explicit application of subgradient descent to the problem with a slightly larger number of iterations. More recently, Cohen \etal \cite{CohenLMPS16} developed one of the fastest $(1+\eps)$-approximation algorithms to date for the geometric median problem via stochastic subgradient methods.

\item[Kernel methods in machine learning]
Kernel functions simulate an inner product space in large or even unbounded dimensions but can be evaluated via simple low dimensional vector operations in the original dimension of input points \cite{ScholkopfS02}. This enables simple 
spherical shape fitting via a smallest enclosing ball algorithm in the high dimensional feature space, which implicitly defines a more complex and even non-convex shape in the original space. The smallest enclosing ball problem in kernel spaces is well-known as the support vector data description (SVDD) by Tax and Duin \cite{TaxD04}. A more subtle connection between (the dual formulations of) several kernel based methods in machine learning and the smallest enclosing ball problem was established by Tsang \etal \cite{TsangKC05}.
\end{description}

%\noindent
\subsection{Our contributions and outline} 
\begin{itemize}
\item We extend in Section \ref{sec:setmed} the geometric median in Euclidean space to the more general set median problem. It consists of finding a center $c\in \REAL^d$ that minimizes the sum of maximum distances to sets of points in a given collection of $N$ point sets. We show how to solve this problem via estimation and sampling techniques combined with a stochastic subgradient descent algorithm.
\item The elements in the collection are sets of up to $n$ points in $\REAL^d$. We discuss in Section \ref{sec:coresets} the possibility of further reducing their size. In the previous work \cite{MunteanuSF14} they were summarized via strong coresets of size $1/\eps^{\Theta(d)}$ for constant dimension $d$. This is not an option in high dimensions where, e.g. $d \approx n$. Reviewing the techniques of Agarwal and Sharathkumar \cite{AgarwalS15} we can show that no reduction below $\min\{n, \exp(d^{1/3})\}$ is possible unless one is willing to sacrifice an additional approximation factor of roughly $\sqrt{2}$. However, we discuss the possibility to achieve roughly a factor $(\sqrt{2}+\eps)$-approximation in streaming via the \emph{blurred-ball-cover} \cite{AgarwalS15} of size $O(1/\eps^3\cdot \log 1/\eps)$, and in an off-line setting via weak coresets \cite{BadoiuC03,BadCla08} of size $O(1/\eps)$.
\item We show in Section \ref{sec:pSEB} how this improves the previously best FPTAS for the probabilistic smallest enclosing ball problem from $O(dn/\eps^3 \cdot \log 1/\eps + 1/\eps^{O(d)})$ to $O(dn/\eps^4 \cdot \log^2 1/\eps)$. In particular the dependence on the dimension $d$ is reduced from exponential to linear and more notably occurs only in distance evaluations between points in $d$-dimensional Euclidean space, but not in the number of sampled points nor in the number of candidate centers to evaluate.
\item This enables in Section \ref{sec:pSVDD} working in very high $D$-dimensional Hilbert spaces whose inner products and distances are given implicitly via positive semidefinite kernel functions. These functions can be evaluated in $O(d)$ time although $D$ is large or even unbounded depending on the kernel function. As an example we extend the well-known support vector data description (SVDD) method to the probabilistic case. SVDD is equivalent to the smallest enclosing ball problem in the implicit high-dimensional feature space.
\end{itemize}

\noindent
%Please find the missing proofs in the full version of this paper \cite{km-pseb-19}.
Please find the missing proofs in Appendix~\ref{append}.

\subsection{General notation}
We denote the set of positive integers up to $n\in\mathbb N$ by $[n]=\{1,\ldots,n\}$. For any vectors $x,y\in\REAL^d$ we denote their inner product by $\langle x,y \rangle = x^T y = \sum\nolimits_{i=1}^{d} x_i y_i$ and the Euclidean norm by $\norm{x}=\sqrt{\langle x,x \rangle} = (\sum\nolimits_{i=1}^{d} x_i^2)^{1/2}$. The Cauchy-Schwarz inequality (CSI) states that $\langle x,y\rangle\leq\norm{x}\norm{y}$.
For any convex function $f\colon \REAL^d \rightarrow \REAL$ we denote by $\partial f(x)=\br{g\in\REAL^d \mid \forall y\in\REAL^d\colon f(x)-f(y) \leq \langle g, x-y \rangle}$ the subdifferential, i.e., the set of subgradients of $f$ at $x$. For any event $\mathcal E$ let the indicator function be $\mathds{1}_{\mathcal E}=1$ if $\mathcal E$ happens and $0$ otherwise. We denote by $B(c,r)=\{x \mid \norm{c-x}\leq r\}$ the Euclidean ball centered at $c\in\REAL^d$ with radius $r\geq 0$. We assume the error parameter satisfies $0<\eps<1/9$. All our results hold with constant probability, say $1/8$, which can be amplified to arbitrary $1-\eta$, $0<\eta<1$, by running $\O{\log 1/\eta}$ independent repetitions and returning the minimum found.

\section{A generalized median problem}
\label{sec:setmed}
The probabilistic smallest enclosing ball (pSEB) problem is to find a center that minimizes the expected maximum distance to points drawn from the input distributions (see Section~\ref{sec:pSEB} and Definition~\ref{def:probseb}), and it can be reduced to two different types of $1$-median problems \cite{MunteanuSF14}. One of them is defined on the set of all non-empty locations in $\REAL^d$ where probabilistic points may appear, equipped with the Euclidean distance. The other is defined on the collection of all possible realizations of probabilistic point sets, and the distance measure between a center $c\in\REAL^d$ and a realization $P_i\subset\REAL^d$ is the maximum distance between the center and any of the realized points, i.e., $\max\nolimits_{p\in P_i} \norm{c-p}$.

We begin our studies with a generalized median problem that we call the \emph{\medgen problem} and covers both of these cases.

\begin{definition}[\medgen problem]
	\label{def:generalmedian}
	Let $\mathcal P=\{P_1,\ldots,P_N\}$ be a family of finite non-empty sets where $\forall i\in[N]\colon P_i \subset \REAL^d$ and $n = \max\{ |P_i|\mid i\in[N] \}$. The \medgen problem on $\mathcal P$ consists in finding a center $c\in\REAL^d$ that minimizes the cost function 
	\[
	f(c)=\sum\limits_{i=1}^N m(c,P_i),
	\]
	 where $m(c,P_i)=\max\nolimits_{p\in P_i} \norm{c-p}$.
\end{definition}

It was noted in \cite{MunteanuSF14} that the distance measure $m$ between any two sets $A,B\subset\REAL^d$ defined by the maximum distance $m(A,B)=\max\nolimits_{a\in A, b\in B} \norm{a-b}$ is not a metric since for any non-singleton set $C\subset \REAL^d$ it holds that $m(C,C)>0$. We stress here that we consider only cases where, as in Definition \ref{def:generalmedian}, one of the sets $A=\lbrace c\rbrace$ is a singleton, and $B=P$ is an arbitrary non-empty set of points from $\REAL^d$. In order to directly apply results from the theory of metric spaces we thus simply define $m(A,B)=0$ whenever $A=B$. 

\begin{lemma}
	\label{lem:metric}
	Let $\mathcal{X}$ be the set of all finite non-empty subsets of $\REAL^d$. We define 
	\[
	m(A,B)=\begin{cases}\max_{a\in A, b\in B} \norm{a-b} & \text{if } A\neq B\\ 
	0 & \text{if } A=B \end{cases}
	\]
	for any $A,B\in \mathcal{X}$. Then $(\mathcal{X},m)$ is a metric space.
\end{lemma}

Note that in case of singleton sets, the \medgen problem in Definition \ref{def:generalmedian} is equivalent to the well-known Fermat-Weber problem also known as $1$-median or geometric median. Also, for $N=1$ it coincides with the smallest enclosing ball or $1$-center problem.

For both of these problems there are known algorithms based on the subgradient method from convex optimization. B\u{a}doiu and Clarkson \cite{BadoiuC03} gave a simple algorithm for approximating the 1-center within $(1+\eps)$-error. Starting from an initial center, it is iteratively moved a little towards the input point that is furthest away. Note that a suitable subgradient at the current center points exactly into the opposite direction. More precisely if $q\in P$ is a point that is furthest away from the current center $c$ then $(c-q)/\norm{c-q}\in \partial \max_{p\in P} \norm{c-p}$. The algorithm can thus be interpreted as a subgradient descent minimizing $\max_{p\in P} \norm{c-p}$. Weiszfeld's algorithm \cite{weiszfeld1937,WeiszfeldP09} was the first that solved the $1$-median problem within additive $O(\eps)$-error in $O(1/\eps)$ subgradient iterations. And indeed one of the fastest algorithms to approximate the geometric median problem to date relies on a stochastic subgradient descent \cite{CohenLMPS16}. The crucial ingredients to turn the additive error to a relative error are finding a suitable starting point that achieves a constant approximation and estimating its initial distance to the optimal solution. Another important step is a bound on the Lipschitz constant of the cost function, see also \cite{beck2014}. We will generalize these approaches to minimizing the cost function $f$ of the \medgen problem from Definition \ref{def:generalmedian}.

First note that $f$ is a convex function which implies that we can apply the theory of convex analysis and optimization. In particular, it implies that the subdifferential $\partial f(c)$ is non-empty for any center $c\in\REAL^d$ and $c$ is locally optimal if and only if $0\in\partial f(c)$. Moreover, any local optimum is also globally optimal by convexity \cite{BoydV2004,nesterov2004}. This implies that if we find a $(1+\eps)$-approximation to a local minimum of the convex function, the convexity implies that it is a $(1+\eps)$-approximation to the global minimum as well.

To see that $f$ is convex, note that the Euclidean norm is a convex function. Therefore the Euclidean distance to some fixed point is a convex function since every translation of a convex function is convex. The maximum of convex functions is a convex function and finally the sum of convex functions is again convex. We prove this claim for completeness.

\begin{lemma}
	\label{lem:fconvex}
The objective function $f$ of the \medgen problem (see Definition \ref{def:generalmedian}) is convex.
\end{lemma}

Next we bound the Lipschitz constant of the function $f$ by $N$. We may get a better bound if we limit the domain of $f$ to a ball of small radius centered at the optimal solution, but the Lipschitz constant cannot be bounded by $o(N)$ in general. We will see later how we can remove the dependence on $N$.

\begin{lemma}
\label{lem:nlipschitz}
The objective function $f$ of the \medgen problem (see Definition \ref{def:generalmedian}) is $N$-Lipschitz continuous, i.e., $|f(x)-f(y)|\leq N\cdot\norm{x-y}$ for all $x,y\in\REAL^d$.
\end{lemma}

We want to minimize $f$ via the subgradient method, see \cite{nesterov2004}. For that sake we need to compute a subgradient $g(c_i)\in \partial f(c_i)$ at the current center $c_i$. To this end we prove the following lemma.

\begin{lemma}
	\label{lem:subgrad1}
	Let $c_i\in \REAL^d$ be any center. For each set $P_j\in\mathcal{P}$, let $p_j\in P_j$ be a point with $\norm{c_i-p_j} = m(c_i,P_j)$. We have
\[
		g(c_i)=\sum\limits_{j=1}^N \frac{c_i-p_j}{\norm{c_i-p_j}} \cdot \mathds{1}_{c_i\neq p_j} \in \partial f(c_i),
\]
	i.e., $g(c_i)$ is a valid subgradient of $f$ at $c_i$.
\end{lemma}

For brevity of presentation we omit the indicator function in any use of the above Lemma in the remainder of this paper and simply define $\left( c_i-p_j\right)/\norm{c_i-p_j}=0$ whenever $c_i=p_j$.

The subgradient computation takes $\O{d{n}N}$ time to calculate, since in each of the $N$ terms of the sum we maximize over $|P_i|\leq {n}$ distances in $d$ dimensions to find a point in $P_i$ that is furthest away from $c$. We are going to discuss the possibility of reducing ${n}$ later. For now we focus on removing the dependence on $N$. To this end we would like to replace the exact subgradient $g(c_i)$ by a uniform sample of only one nonzero term which points into the right direction in expectation. We formalize this in the following lemma.

\begin{lemma}
\label{lem:expgradient}
Let $c_i\in \REAL^d$ be any fixed center. For each set $P_j\in\mathcal{P}$, let $p_j\in P_j$ be a point with $\norm{c_i-p_j} = m(c_i,P_j)$. Let $\tilde g(c_i)$ be a random vector that takes the value $\tilde{g}(c_i)=\left( c_i-p_j \right) / \norm{c_i-p_j}$ for $j\in[N]$ with probability $1/N$ each. Then $\Ex{\norm{\tilde{g}(c_i)}^2}\leq 1$ and $\Ex{\tilde{g}(c_i)} = {g}(c_i)/N$, where ${g}(c_i)\in \partial f(c_i)$ is the subgradient given in Lemma \ref{lem:subgrad1}.
\end{lemma}

We can now adapt the deterministic subgradient method from \cite{nesterov2004} using the random unbiased subgradient of Lemma \ref{lem:expgradient} in such a way that the result is in expectation a $(1+\eps)$-approximation to the optimal solution. This method is presented in Algorithm \ref{alg:subgradient}. Given an initial center $c_0$, a fixed step size $s$, and a number of iterations $\ell$, the Algorithm iteratively picks a set $P_j\in\mathcal P$ uniformly at random and chooses a point $p_j\in P_j$ that attains the maximum distance to the current center. This point is used to compute an approximate subgradient via Lemma \ref{lem:expgradient}. The Algorithm finally outputs the best center found in all iterations.

\begin{algorithm}[ht!]
\KwData{A family of non-empty sets $\mathcal P=\br{P_1,\ldots,P_N}$, where $P_i\subset\REAL^d$}
\KwResult{A center $\tilde c \in \REAL^d$}
\caption{Stochastic subgradient method}\label{alg:subgradient}
	Determine an initial center $c_0$\;
	Fix a step size $s$\;
	Fix the number of iterations $\ell$\;
	\For{$i\leftarrow 1$ \KwTo $\ell$}{
		Choose an index $j\in\left[ N\right]$ uniformly at random and compute $\tilde{g}(c_{i-1})$, cf. Lemma \ref{lem:expgradient}\;
		$c_i=c_{i-1} - s \cdot \tilde{g}(c_{i-1})$\;
	}  
	\Return $\tilde{c} \in \operatorname{argmin}\limits_{c\in\{c_i\mid i=0,\ldots,\ell\}} f(c)$\;
\end{algorithm}

The following theorem bounds in expectation the quality of the output that our subgradient algorithm returns. It is a probabilistic adaptation of a result in convex optimization \cite{nesterov2004}.

\begin{theorem}
\label{thm:gradboundprob}
Consider Algorithm~\ref{alg:subgradient} on input $\mathcal P=\br{P_1,\ldots,P_N}$ for the \medgen problem with objective function $f$, see Definition \ref{def:generalmedian}. Let $c^*\in \operatorname{argmin}\limits_{c\in\REAL^d} f(c)$. Let $R=\norm{c_0 - c^*}$. Then
\begin{equation*}
%\label{eqn:goalbound}
\Exp{\tilde{c}}{ f(\tilde{c}) -f\left(c^*\right) } \leq N \cdot \frac{R^2 + (\ell+1) s^2}{2(\ell+1)s},
\end{equation*}
where the expectation is taken over %all possible positions of
the random variable $\tilde{c} \in \operatorname{argmin}\limits_{c\in\{c_i\mid i=0,\ldots,\ell\}} f\left(c\right)$, i.e., the output of Algorithm \ref{alg:subgradient}.
\end{theorem}

Our aim now is to choose the parameters $\ell,s,$ and $c_0$ of Algorithm~\ref{alg:subgradient} in such a way that the bound given in Theorem~\ref{thm:gradboundprob} becomes at most $\eps f(c^*)$. 
To this end we can choose the initial center $c_0$ and bound its initial distance $R=\norm{c_0-c^*}$ proportional to the average cost $\O{f(c^*)/N}$ with constant probability using a simple Markov argument.
\begin{lemma}
\label{lem:startingpoint}
Choose a set $P$ from $\mathcal P = \{P_1,\ldots,P_N\}$ uniformly at random and let $c_0$ be an arbitrary point of $P$. Then for any constant $0<\delta_1<1$ it holds that $R=\norm{c_0-c^*}\leq f(c^*)/(\delta_1 N)$ with probability at least $1-\delta_1$.
\end{lemma}

Assume that we know the value of $R$, and set the step size to $s=R/\sqrt{\ell+1}$, then Theorem~\ref{thm:gradboundprob} and Lemma \ref{lem:startingpoint} imply that for some constant $C$
\[
%\label{eqn:bound1}
\Exp{\tilde c}{ f(\tilde{c}) - f\left(c^*\right) }
\leq \frac{NR}{\sqrt{\ell+1}}\leq \frac{Cf(c^*)}{\sqrt{\ell+1}}
\]
holds with constant probability. We thus only need to run the algorithm for $\ell\in\O{1/\eps^2}$ iterations to get within $\eps f(c^*)$ error. But choosing this particular step size requires to know the optimal center in advance. To get around this, we attempt to estimate the average cost. More formally, we are interested in a constant factor approximation of $f(c^*)/N$. It turns out that we can do this based on a small sample of the input sets unless our initial center is already a good approximation. But in the latter case we do not care about all the subsequent steps or step sizes, since we are already done after the initialization. The proof technique is originally from \cite{KumarSS10} and is adapted here to work in our setting with sets of points.

\begin{lemma}
\label{lem:Restimate}
There exists an algorithm that based on a sample $\mathcal S\subseteq \mathcal{P}$ of size $|\mathcal S|=1/\eps$ returns an estimate $\tilde{R}$ and an initial center $c_0$ in time $O(d{n}/\eps)$ such that with constant probability one of the following holds:
\begin{enumerate}
	\item[a)] $\eps f(c^*)/N \leq \tilde{R} \leq \left( 2/\eps^3\right)\cdot f(c^*)/N$ and $\norm{c_0-c^*}\leq 8 f(c^*)/N$;
	\item[b)] $\norm{c_0-c^*}\leq 4 \eps f(c^*)/N$.
\end{enumerate}
\end{lemma}

Lemma \ref{lem:Restimate} has the following consequence. Either the initial center $c_0$ is already a $(1+4\eps)$-approximation, in which case we are done. Or we are close enough to an optimal solution and have a good estimate on the step size to find a $(1+4\eps)$-approximation in a constant number of iterations.

Another issue that we need to take care of, is finding the best center in the last line of Algorithm \ref{alg:subgradient} efficiently. We cannot do this exactly since evaluating the cost even for one single center takes time $O(d{n}N)$. However, we can find a point that is a $(1+\eps)$-approximation of the best center in a finite set of candidate centers using a result from the theory of discrete metric spaces.

To this end we can apply our next theorem which is originally due to Indyk and Thorup in \cite{IndykDiss, Thorup04} and adapted here to work in our setting. The main difference is that in the original work the set of input points and the set of candidate solutions are identical. In our setting, however, we have that the collection of input sets and the set of candidate solutions may be completely distinct and the distance measure is the maximum distance (see Lemma~\ref{lem:metric}).
\begin{theorem}
\label{thm:thorup}
Let $\mathcal{Q}$ be a set of uniform samples with repetition from $\mathcal{P}$. Let $\mathcal{C}$ be a set of candidate solutions. Let $a\in \mathcal{C}$ minimize $\sum_{Q\in \mathcal{Q}} m(a,Q)$ and let $\hat c=\operatorname{argmin}_{c\in\mathcal C} f(c)$. Then 
\[
\Pr{\sum\nolimits_{{P}\in\mathcal{P}} m\left(a,P\right) {>} \left(1+\eps\right) \sum\nolimits_{{P}\in\mathcal{P}} m\left(\hat c,P\right)} \leq |\mathcal C| \cdot e^{-\eps^2 |\mathcal Q|/64}.
\]
\end{theorem}
Putting all pieces together we have the following Theorem.
\begin{theorem}
\label{thm:mainresult}
Consider an input $\mathcal{P}=\lbrace P_1,\ldots, P_N\rbrace$, where for every $i\in[N]$ we have $P_i\subset \REAL^d$ and $n = \max\{ |P_i|\mid i\in[N] \}$. There exists an algorithm that computes a center $\tilde c$ that is with constant probability a $(1+\eps)$-approximation to the optimal solution $c^*$ of the \medgen problem (see Definition~\ref{def:generalmedian}). Its running time is $\O{d{n}/\eps^4 \cdot \log^2 1/\eps}$.
\end{theorem}

\subsection{On reducing the size of the input sets}
\label{sec:coresets}
We would like to remove additionally the linear dependence on {$n$} for the maximum distance computations. This is motivated from the streaming extension of \cite{MunteanuSF14} where one aims at reading the input once in linear time and all subsequent computations should be sublinear, or preferably independent of $n$. To this end a grid based strong coreset of size $1/\eps^{\Theta(d)}$ was used. However, here we focus on reducing the dependence on $d$, and exponential is not an option if we want to work in high dimensions. It turns out that without introducing an exponential dependence on $d$, we would have to lose a constant approximation factor. Pagh \etal \cite{PaghSSS17} showed that if the coreset is a subset of the input and approximates furthest neighbor queries to within less than a factor of roughly $\sqrt{2}$, then it must consist of $\Omega\left(\min\lbrace n, \exp (d^{1/3})\rbrace \right)$ points. This was shown via a carefully constructed input point set of Agarwal and Sharathkumar \cite{AgarwalS15} who used it to prove lower bounds on streaming algorithms for several extent problems. 

In the next theorem, we review the techniques of the latter reference to show a slightly stronger result, namely no small data structure can exist for answering maximum distance queries to within a factor of less than roughly $\sqrt{2}$. In comparison to the previous results \cite{AgarwalS15,PaghSSS17}, it is not limited to the streaming setting, and it is not restricted to subsets of the input.

\begin{theorem}
	\label{lem:datastructurelower}
	Any data structure that, with probability at least $2/3$, $\alpha$-approximates maximum distance queries on a set $S\subset \REAL^d$ of size $|S|=n$, for $\alpha <\sqrt{2} \left( 1-2/d^{1/3}\right)$, requires $\Omega\left(\min\lbrace n, \exp \left(d^{1/3}\right)\rbrace \right)$ bits of storage.
\end{theorem}

On the positive side it was shown by Goel \etal \cite{GoelIV01} that a $\sqrt{2}$-approximate furthest neighbor to the point $c\in\REAL^d$ can always be found on the surface of the smallest enclosing ball of the sets $P_i$, using linear preprocessing time $\tilde{O}\left( dn \right)$ and $\tilde{O}\left( d^2\right)$ query time. Thus, if we plug in the coresets of B\u{a}doiu and Clarkson \cite{BadoiuC03,BadCla08} of size $O(1/\eps)$ instead of the entire sets $P_i\in\mathcal{P}$ to evaluate $m(c,P_i)$, we would have a sublinear time algorithm (in $n$) after reading the input via a $\sqrt{2}\left( 1+\eps\right)$-approximation to any query $m(c,P_i)$. In a streaming setting the same bound can be achieved via the \emph{blurred-ball-cover} of Agarwal and Sharathkumar \cite{AgarwalS15} of slightly larger size $\O{1/\eps^3 \cdot\log 1/\eps}$. Otherwise, Goel \etal \cite{GoelIV01} have shown that using $\O{dn^{1+1/(1+\eps)}}$ preprocessing time and $\tilde{O}\left( dn^{1/(1+\eps)}\right)$ query time, one can obtain a $(1+\eps)$-approximation for the furthest neighbor problem. Note that in this case the preprocessing time is already superlinear in $n$, and in particular the exponent is already larger than $1.7$ for $1+\eps<\sqrt{2}$.

\section{Applications}
\subsection{Probabilistic smallest enclosing ball}
\label{sec:pSEB}
We apply our result to the probabilistic smallest enclosing ball problem, as given in \cite{MunteanuSF14}. In such a setting, the input is a set $\mathcal{D}=\lbrace D_1,\ldots, D_n\rbrace$ of $n$ discrete and independent probability distributions. The $i$-th distribution $D_i$ is defined over a set of $z$ possible locations $q_{i,j}\in \REAL^d \cup \lbrace \bot \rbrace$, for $j\in \left[ z\right]$, where $\bot$ indicates that the $i$-th point is not present in a sampled set, i.e., $q_{i,j}=\bot \Leftrightarrow \lbrace q_{i,j}\rbrace = \emptyset$. We call these points \emph{probabilistic points}. Each location $q_{i,j}$ is associated with the probability $p_{i,j}$, such that $\sum_{j=1}^z p_{i,j} =1$, for every $i\in \left[ n \right]$. Thus the probabilistic points can be considered as independent random variables $X_i$.

A probabilistic set $X$ consisting of probabilistic points is also a random variable, where for each random choice of indices $\left( j_1,\ldots,j_n\right) \in \left[z\right]^n$ there is a realization $P_{\left( j_1,\ldots,j_n\right)} = X\left( j_1,\ldots,j_n\right)= \left( q_{1,j_1},\ldots, q_{n,j_n}\right)$. By independence of the distributions $D_i$, $i\in \left[ n\right]$, it holds that $\Pr{X=P_{\left( j_1,\ldots,j_n\right)}} = \prod_{i=1}^n p_{i,j_i}$.

The probabilistic smallest enclosing ball problem is defined as follows. Here we may assume that the distance of any point $c\in \REAL^d$ to the empty set is 0.
\begin{definition}{(\cite{MunteanuSF14})}
\label{def:probseb}
	Let $\mathcal{D}$ be a set of $n$ discrete distributions, where each distribution is defined over $z$ locations in $\REAL^d\cup\br{\bot}$.
	The probabilistic smallest enclosing ball problem is to find a center $c^*\in\REAL^d$ that minimizes the expected smallest enclosing ball cost, i.e., \[ c^* \in \operatorname{argmin}_{c\in\mathbb{R}^d} \Exp{X}{m(c,X)}, \]
	where the expectation is taken over the randomness of $X\sim \mathcal D$.
\end{definition}

The authors of \cite{MunteanuSF14} showed a reduction of  
the probabilistic smallest enclosing ball problem to computing a 
solution for the \medgen problem of Definition~\ref{def:generalmedian}. Their algorithm distinguishes between two cases. In the first case the probability of obtaining a nonempty realization $P\neq\emptyset$ is small, more formally $\sum_{q_{i,j}\in {Q}} p_{i,j} \leq \eps$, {where $Q = \lbrace q_{i,j}\mid q_{i,j}\neq\bot, i\in [n], j\in [z]\rbrace$}, and thus we have little chance of gaining information by sampling realizations. However, it was shown that 
\begin{align*}
(1-\eps)\cdot \Exp{X}{\sum\limits_{p\in X} \norm{c -p} } \leq \Exp{X}{m(c,X)} \leq \Exp{X}{\sum\limits_{p\in X} \norm{c -p} },% = \sum_{i,j} p_{i,j}\cdot \norm{c'-q_{i,j}}, 
\end{align*}
where $\Exp{X}{\sum_{p\in X} \norm{c -p} } = \sum_{i,j} p_{i,j}\cdot \norm{c-q_{i,j}}$ is a weighted version of the deterministic 1-median problem, cf. \cite{CormodeM08}, and thus also a weighted instance of the \medgen problem. In the second case, the probability that a realization contains at least one point is reasonably large. Therefore by definition of the expected value and $m(c,\emptyset)=0$ we have 
\[
\Exp{X}{m(c,X)} = \sum\nolimits_{P\neq \emptyset} \Pr{X=P}\cdot m(c,P),
\] which is a weighted version of the \medgen problem. Depending on these two cases, we sample a number of elements, non-empty locations or non-empty realizations, and solve the resulting \medgen problem using the samples in Theorem \ref{thm:mainresult} {for computing the approximate subgradients}.

Algorithm~\ref{alg:SEB} adapts this framework. It differs mainly in three points from the previous algorithm of \cite{MunteanuSF14}. First, the number of samples had a dependence on $d$ hidden in the $O$-notation. This is not the case any more. Second, the sampled realizations are not sketched via coresets of size $1/\eps^{\Theta(d)}$ any more, as discussed in Section \ref{sec:coresets}. Third, the running time of the actual optimization task is reduced via Theorem \ref{thm:mainresult} instead of an exhaustive grid search.

%TODO cam-ready: line numbers + remove '-' as line starters
\begin{algorithm}[ht]
\KwData{A set $\mathcal{D}$ of $n$ point distributions over {$z$} locations in $\REAL^d$, a parameter $\eps<1/9$}
\KwResult{A center $\hat c \in \REAL^d$}
\caption{Probabilistic smallest enclosing ball}\label{alg:SEB}
	$Q\leftarrow \lbrace q_{i,j}\mid q_{i,j}\neq\bot, i\in [n], j\in [z]\rbrace$ \tcc*[r]{the set of non-empty locations}
	Set a sample size $k\in \O{1/\eps^2 \log (1/\eps)}$\;
	\If{$\sum_{q_{i,j}\in Q} p_{i,j} \leq \eps$}{
		- Pick a random sample $R$ of $k$ locations from $\mathcal P = Q$, where for every $r\in R$ we have $r=q_{ij}$ with probability proportional to $p_{ij}$\;
		- Compute $\hat c\in \REAL^d$ that is a $(1+\eps)$-approximation 
		using the sampled points $R$ one-by-one for computing the approximate {sub}gradients in the algorithm of Theorem~\ref{thm:mainresult}\;
	}
	\Else{
		- Sample a set $R$ of $k$ non-empty realizations from the input distributions $\mathcal{D}$\;
		- Compute $\hat c\in \REAL^d$ that is a $(1+\eps)$-approximation %to $\operatorname{argmin}_{c'\in\REAL^d} f(c')$, 
		using the sampled realizations $R$ one-by-one for computing the approximate {sub}gradients in the algorithm of Theorem~\ref{thm:mainresult}\;
	}
	\Return $\hat c$\;
\end{algorithm}

\begin{theorem}
\label{thm:SEBalgorithm}
Let $\mathcal{D}$ be a set of $n$ discrete distributions, where each distribution is defined over $z$ locations in $\REAL^d \cup \lbrace\bot\rbrace$. Let $\tilde c\in\REAL^d$ denote the output of Algorithm~\ref{alg:SEB} on input $\mathcal{D}$, and {let} the approximation parameter be $\eps<1/9$. Then with constant probability the output is a $(1+\eps)$-approximation for the probabilistic smallest enclosing ball problem. I.e., it holds that 
\[
\Exp{X}{m(\tilde c,X)} \leq (1+\eps) \min\nolimits_{c\in \REAL^d} \Exp{X}{m(c,X)}.
\]
The running time of Algorithm~\ref{alg:SEB} is $\O{dn\cdot ( z/\eps^3 \cdot \log 1/\eps +  1/\eps^4 \cdot \log^2 1/\eps)}$.
\end{theorem}

Comparing to the result of \cite{MunteanuSF14}, the running time is reduced from $O(dnz/\eps^{O(1)} + 1/\eps^{O(d)})$ to $O(dnz/\eps^{O(1)})$, i.e., our dependence on the dimension $d$ is no longer exponential but only linear. Note, in particular, that the factor of $d$ plays a role only in computations of distances between two points in $\REAL^d$. Further the sample size and the number of centers that need to be evaluated do not depend on the dimension $d$ any more. This will be crucial in the next application.

\subsection{Probabilistic support vector data description}
\label{sec:pSVDD}
Now we turn our attention to the support vector data description (SVDD) problem \cite{TaxD04} and show how to extend it to its probabilistic version. To this end, let $K\colon \REAL^d \times \REAL^d \rightarrow \REAL$ be a positive semidefinite kernel function. It is well known by Mercer's theorem, cf. \cite{ScholkopfS02}, that such a function implicitly defines the inner product of a high dimensional Hilbert space $\mathcal{H}$, say $\REAL^D$ where $D \gg d$. This means we have $K(x,y)=\langle \varphi(x), \varphi(y)\rangle$, where $\varphi\colon\REAL^d\rightarrow \mathcal H$ is the so called feature mapping associated with the kernel. Examples for such kernel functions include polynomial transformations of the standard inner product in $\REAL^d$ such as the constant, linear or higher order polynomial kernels. In these cases $D$ remains bounded but grows as a function of $d$ raised to the power of the polynomials' degree. Other examples are the exponential, squared exponential, Mat\'ern, or rational quadratic kernels, which are transformations of the Euclidean distance between the two low dimensional vectors. The dimension $D$ of their implicit feature space is in principle unbounded. Despite the large dimension $D\gg d$, all these kernels can be evaluated in time $O(d)$  \cite{RasmussenW06}.

It is known that the SVDD problem is equivalent to the smallest enclosing ball problem in the feature space induced by the kernel function \cite{TsangKC05}, i.e., given the kernel $K$ with implicit feature mapping $\varphi\colon \REAL^d \rightarrow \mathcal H$, and an input set $P\subseteq \REAL^d$, the task is to find 
\begin{equation}
	\label{dSVDD}
	c^*\in \operatorname{argmin}_{c \in \mathcal H} \max\nolimits_{p\in P}\norm{c-\varphi(p)} = \operatorname{argmin}_{c \in \mathcal H} m(c,\varphi(P)),
\end{equation}
where $\varphi(P)=\{\varphi(p)\mid p\in P\}$.

Now we extend this to the probabilistic setting as we did in Section \ref{sec:pSEB}. The input is again a set $\mathcal{D}$ of $n$ discrete and independent probability distributions, where $D_i\in\mathcal{D}$ is defined over a set of $z$ locations $q_{i,j}\in \REAL^d \cup\lbrace \bot\rbrace$. Note that the mapping $\varphi$ maps the locations $q_{i,j}$ from $\REAL^d$ to $\varphi(q_{i,j})$ in $\mathcal{H}$, and we assume $\varphi(\bot)=\bot$. Then the probabilistic SVDD problem is given by the following adaptation of Definition~\ref{def:probseb}.
\begin{definition}
	\label{def:probsvdd}
	Let $\mathcal{D}$ be a set of $n$ discrete distributions, where each distribution is defined over $z$ locations in $\REAL^d\cup\br{\bot}$. Let $K\colon \REAL^d \times \REAL^d \rightarrow \REAL$ be a kernel function with associated feature map $\varphi\colon \REAL^d \rightarrow \mathcal H$.
	The probabilistic support vector data description (pSVDD) problem is to find a center $c^*\in\mathcal H$ that minimizes the expected SVDD cost, i.e., \[ c^* \in \operatorname{argmin}_{c\in\mathcal H} \Exp{X}{m(c,\varphi(X))}, \]
	where the expectation is taken over the randomness of $X\sim \mathcal D$.
\end{definition}
Note that the deterministic problem, see Equation (\ref{dSVDD}), is often stated with squared distances \cite{TaxD04,TsangKC05}. It does not matter whether we minimize the maximum distance or any of its powers or any other monotone transformation. In the probabilistic case this is not true. Consider for instance squared distances. The resulting problem would be similar to a $1$-means rather than a $1$-median problem. Huang \etal \cite{HuangLPW16} have observed that minimizing the expected maximum squared distance corresponds to minimizing the expected area of an enclosing ball in $\REAL^2$. This observation can be generalized to the expected volume of an enclosing ball in $\REAL^p$ when the $p$-th powers of distances are considered. Considering $p=2$ might also have advantages when dealing with Gaussian input distributions due to their strong connection to squared Euclidean distances. In a general setting of the probabilistic smallest enclosing ball problem however, it is natural to minimize in expectation the maximum Euclidean distance, as in Definitions~\ref{def:probseb} and \ref{def:probsvdd}, since its radius is the primal variable to minimize.

Next we want to show how to find a $(1+\eps)$-approximation for the pSVDD problem. Explicitly computing any center $c\in\mathcal H$ takes $\Omega(D)$ time and space which is prohibitive not only when $D=\infty$. Note that for the SEB and SVDD problems, any reasonable center lies in the convex hull of the input points. Since taking the expectation is simply another linear combination over such centers, we can express any center $c\in\mathcal H$ as a linear combination of the set of non-empty locations, i.e., 
\[
	c=\sum\nolimits_{q_{u,v}\in Q} \gamma_{u,v} \varphi(q_{u,v})
\]
The idea is to exploit this characterization to simulate Algorithm~\ref{alg:subgradient} and thereby Algorithm~\ref{alg:SEB} to work in the feature space $\mathcal{H}$ by computing the centers and distances only implicitly.

	For now, assume that any distance computation can be determined. Note that sampling a set $P_i \subset \REAL^d$ is the same as sampling the set $\varphi(P_i)$ of corresponding points in $\mathcal{H}$ from the same distribution. We assume that we have a set of locations or realizations $\mathcal{P}=\lbrace P_1,\ldots, P_N\rbrace$, with $P_i\subset\REAL^d$. The remaining steps are passed to Theorem \ref{thm:mainresult} which is based on Algorithm~\ref{alg:subgradient}. First, we show the invariant that each center $c_i$ reached during its calls to Algorithm~\ref{alg:subgradient} can be updated such that we maintain a linear combination $c_i=\sum_{u,v} \gamma_{u,v} \varphi(q_{u,v})$, where at most $i+1$ terms have $\gamma_{u,v}\neq 0$.
\begin{itemize}
\item The initial center $c_0 \in \mathcal{H}$ is chosen by sampling uniformly at random a set $P\in \mathcal{P}$ via Lemma~\ref{lem:startingpoint}. We take any point $q\in P$, $q\neq \bot$, which maps to $c_0=\varphi(q)$. Thus the invariant is satisfied at the beginning, where the corresponding coefficient is $\gamma=1$ and all other coefficients are zero.

\item In each iteration we randomly sample a set $P \in \mathcal{P}$ to simulate the approximate subgradient $\tilde{g}\left(c_{i}\right)$ at the current point $c_{i}$. The vector $\tilde{g}\left(c_{i}\right)$ is a vector between $c_{i}$ and some point $p_{j,k}=\varphi(q_{j,k})\in \mathcal{H}$, such that $q_{j,k}$ maximizes $\norm{c_i-\varphi(q')}$ over all $q'\in P$ (cf. Lemma~\ref{lem:expgradient}). 

To implicitly update to the next center $c_{i+1}$ note that (cf. Algorithm~\ref{alg:subgradient})
\[
c_{i+1} =  c_{i} - s \cdot \frac{c_i-\varphi(q_{j,k})}{\norm{c_i-\varphi(q_{j,k})}} = \left(1-\frac{s}{\norm{c_i-\varphi(q_{j,k})}}\right)\cdot c_{i} + \frac{s}{\norm{c_i-\varphi(q_{j,k})}} \cdot\varphi(q_{j,k}) .
\]
Assume the invariant was valid that $c_i$ was represented as $c_i=\sum_{u,v} \gamma_{u,v} \varphi(q_{u,v})$ with at most $i+1$ non-zero coefficients. Then it also holds for the point $c_{i+1}$ since the previous non-zero coefficients of $\varphi(q_{u,v})$ are multiplied by $1-s/\norm{c_i-\varphi(q_{j,k})}$ and the newly added $\varphi(q_{j,k})$ is assigned the coefficient $s/\norm{c_i-\varphi(q_{j,k})}$. So there are at most $i+2$ non-zero coefficients.
\end{itemize}
Therefore, we do not have to store the points $c_i$ explicitly while performing Algorithm~\ref{alg:subgradient}. The implicit representation can be maintained via a list storing points that appear in the approximate subgradients and their corresponding non-zero coefficients.

To actually compute the coefficients, we need to be able to compute Euclidean distances as well as determine $s$. Using Lemma \ref{lem:Restimate} we determined the step size $s$ via an estimator $\tilde R=\sum_{P\in \mathcal{S}} m(c_0,P)$ based on a small sample $\mathcal S$. In particular this requires distance computations again. To this end, we show how to compute $\norm{c_i-\varphi(q)}$ for any location $q\in\REAL^d$. Recall that the kernel function implicitly defines the inner product in $\mathcal{H}$. It thus holds that
\begin{align}
	\label{eqn:dist_implicit} \nonumber
\norm{c_i-\varphi(q)}^2 &= \norm{\sum\nolimits_{u,v} \gamma_{u,v}\varphi(q_{u,v}) -\varphi(q)}^2 = \norm{\sum\nolimits_{w=0}^i \gamma_{w}\varphi(q_{w}) -\varphi(q)}^2 \\ \nonumber
& = \norm{\sum\nolimits_{w=0}^i \gamma_{w}\varphi(q_{w})}^2 +\norm{\varphi(q)}^2 - 2\sum\nolimits_{w=0}^i \gamma_w \left\langle \varphi(q_w), \varphi(q) \right\rangle \\ 
& = \sum\nolimits_{w = 0}^{i}\sum\nolimits_{w' = 0}^{i} %\sum\nolimits_{m=0}^i
\gamma_w\gamma_{w'} K(q_w,q_{w'}) + K(q,q) -2\sum\nolimits_{w=0}^i \gamma_w K(q_w,q),
\end{align}
where $w,w'\in\lbrace 0,\ldots,i\rbrace$ index the locations $q_w,q_{w'}$ with corresponding $\gamma_w,\gamma_{w'} \neq 0$ in iteration $i$.
Therefore, we have the following Theorem.
\begin{theorem}
	\label{thm:algsvdd}
	Let $\mathcal{D}$ be a set of $n$ discrete distributions, where each distribution is defined over $z$ locations in $\REAL^d \cup \lbrace\bot\rbrace$. There exists an algorithm that implicitly computes $\tilde c\in\mathcal H$ that with constant probability is a $(1+\eps)$-approximation for the probabilistic support vector data description problem. I.e., it holds that 
	\[
	\Exp{X}{m(\tilde c,\varphi(X))} \leq (1+\eps) \min\nolimits_{c\in \mathcal H} \Exp{X}{m(c,\varphi(X))},
	\]
	where the expectation is taken over the randomness of $X\sim\mathcal D$.
	The running time of the algorithm is $\O{dn\cdot \left( z/\eps^3 \cdot \log 1/\eps +  1/\eps^8 \cdot \log^2 1/\eps\right)}$.
\end{theorem}

\section{Conclusion and open problems}
We studied the \medgen problem in high dimensions that minimizes the sum of maximum distances to the furthest point in each input set. We presented a $(1+\varepsilon)$-approximation algorithm whose running time is linear in $d$ and independent of the number of input sets. We further discussed that in high dimensions the size of the input sets cannot be reduced sublinearly without losing a factor of roughly $\sqrt{2}$.
Our work resolves an open problem of \cite{MunteanuSF14} and improves the previously best algorithm for the probabilistic smallest enclosing ball problem in high dimensions by reducing the dependence on $d$ from exponential to linear. This enables running the algorithm in high dimensional Hilbert spaces induced by kernel functions which makes it more flexible and viable as a building block in machine learning and data analysis. As an example we transferred the kernel based SVDD problem of \cite{TaxD04} to the probabilistic data setting.
Our algorithms assume discrete input distributions. It would be interesting to extend them to various continuous distributions. The pSEB problem minimizes the expected maximum distance. When it comes to minimizing volumes of balls or in the context of Gaussian distributions it might be interesting to study higher moments of this variable. This corresponds to a generalization of the \medgen problem to minimizing the sum of higher powers of maximum distances. Finally we hope that our methods may help to extend more shape fitting and machine learning problems to the probabilistic setting.

%%
%% Bibliography
%%

\bibliography{prob1med}

\appendix

\section{Omitted proofs}
\label{append}

\begin{proof}[Proof of Lemma \ref{lem:metric}]
	The non-negativity and symmetry properties of $m$ follow from the corresponding metric properties in Euclidean space $(\REAL^d,\norm{\cdot})$ and by definition. If $A=B$, then $m(A,B)=0$ holds by definition. Otherwise there exist elements $a\in A$, $b\in B$, $a\neq b$, and thus $m(A,B)\geq \norm{a-b}>0$. This proves the identity of indiscernible elements.
	
	To prove the validity of the triangle inequality, let $A,B,C\in\mathcal{X}$ be distinct. Let $a\in A$, $c\in C$ be points such that $m(A,C)=\norm{a-c}$. For any $b\in B$ it holds that 
	\[
	m(A,C)=\norm{a-c} \leq \norm{a-b} + \norm{b-c} \leq m(A,B)+m(B,C),
	\]
	using triangle inequality in $(\REAL^d,\norm{\cdot})$ and the definition of $m$. Now consider the cases where at least two sets are equal. In the case that $A=C$ the claim follows from the non-negativity property. In the case that $A=B$, we have
	\[
	m(A,C)= 0+ m(A,C) \leq m(A,B)+m(B,C).
	\]
	The case $B=C$ is analogous.
\end{proof}

\begin{proof}[Proof of Lemma \ref{lem:fconvex}]
	Let $x,y\in\REAL^d$, let $\lambda_1\in[0,1], \lambda_2=1-\lambda_1$, and let $p_i^*$ maximize $\norm{\lambda_1 x + \lambda_2 y-p_i}$ over all $p_i \in P_i$. Then \begin{align*}
		f(\lambda_1 x + \lambda_2 y) &= \sum\limits_{i=1}^N \max\limits_{p_i\in P_i}\norm{\lambda_1 x + \lambda_2 y - p_i}
		= \sum\limits_{i=1}^N \norm{\lambda_1 x + \lambda_2 y - \underbrace{(\lambda_1+\lambda_2)}_{=1}p_i^*}\\
		&= \sum\limits_{i=1}^N \norm{\lambda_1 (x-p_i^*) + \lambda_2 (y - p_i^*)}
		\leq \lambda_1 \sum\limits_{i=1}^N \norm{x-p_i^*} + \lambda_2 \sum\limits_{i=1}^N \norm{y - p_i^*}\\
		&\leq \lambda_1 \sum\limits_{i=1}^N \max\limits_{p_i\in P_i}\norm{x-p_i} + \lambda_2 \sum\limits_{i=1}^N \max\limits_{p_i\in P_i}\norm{y - p_i}
		= \lambda_1 f(x) + \lambda_2 f(y)
	\end{align*}
	follows as we have claimed.
\end{proof}

\begin{proof}[Proof of Lemma \ref{lem:nlipschitz}]
	Fix any $x,y\in\REAL^d$. Let $p_i^* \in \operatorname{argmax}\nolimits_{p_i\in P_i}\norm{x-p_i}$. By definition of $f$ and applying the triangle inequality to every single term we have
	\begin{align*}
		|f(x)-f(y)| &= \left|\sum\limits_{i=1}^N m(x,P_i) - \sum\limits_{i=1}^N m(y,P_i)\right| \leq \sum\limits_{i=1}^N \left|m(x,P_i) - m(y,P_i)\right| \\
		%&\leq \sum\limits_{i=1}^N \left|\max\limits_{p_i\in P_i} \norm{x-p_i}  - \max\limits_{p_i\in P_i} \norm{y-p_i}\right|\leq \sum\limits_{i=1}^N \left| \norm{x-p_i^*} - \norm{y-p_i^*}\right|\\
		&\leq \sum\limits_{i=1}^N m(x,y) = N \norm{x-y}. \qedhere
	\end{align*}
\end{proof}

\begin{proof}[Proof of Lemma \ref{lem:subgrad1}]
	Let $c^*\in\operatorname{argmin}_{c\in\REAL^d} f(c)$. We first prove that for each term $j\in[N]$
	\begin{align}
		\label{eqn:subgrad1}
		\left\langle \frac{ c_i-p_j }{\norm{c_i-p_j} }\cdot \mathds{1}_{c_i\neq p_j}, c_i-c^* \right\rangle \geq m(c_i,P_j) - m(c^*,P_j).
	\end{align}
	Suppose $c_i=p_j$, then $\left\langle 0, c_i-c^* \right\rangle = 0 \geq 0 - m(c^*,P_j) = m(c_i,P_j) - m(c^*,P_j)$. Otherwise let $p_j^*\in P_j$ be a point such that $\norm{c^*-p_j^*}=m(c^*,P_j)$. We have 
	\begin{align*}
		&\phantom{\;\;=\;\;} \frac{\left\langle c_i-p_j , c_i-c^* \right\rangle}{\norm{c_i-p_j}}
		= \frac{\left\langle c_i-p_j , c_i-p_j+p_j-c^* \right\rangle}{\norm{c_i-p_j}} \\
		&\;=\; \frac{\left\langle c_i-p_j , c_i-p_j-(c^*-p_j) \right\rangle}{\norm{c_i-p_j}} 
		= \frac{ \left\langle c_i-p_j , c_i-p_j \right\rangle - \left\langle c_i-p_j , c^*-p_j \right\rangle }{\norm{c_i-p_j}} \\
		&\stackrel{\text{CSI}}{\geq} \frac{\norm{c_i-p_j}^2 - \norm{c_i-p_j}\norm{c^*-p_j} }{\norm{c_i-p_j}} \geq \norm{c_i-p_j} - \norm{c^*-p_j^*} = m(c_i, P_j) - m(c^*, P_j),
	\end{align*}
	which follows by Cauchy-Schwarz inequality (CSI) and maximality of $p_j^*$.
	Now summing Equation (\ref{eqn:subgrad1}) over all $j\in [ N ]$ we have 
	\begin{align*}
		\left\langle \sum\limits_{j=1}^N \frac{c_i-p_j}{\norm{c_i-p_j}}\cdot \mathds{1}_{c_i\neq p_j}, c_i-c^*\right\rangle
		&= \sum\limits_{j=1}^N \left\langle \frac{c_i-p_j}{\norm{c_i-p_j}}\cdot \mathds{1}_{c_i\neq p_j}, c_i-c^*\right\rangle \\
		&\geq \sum\limits_{j=1}^N \left( m(c_i,P_j) - m(c^*,P_j) \right) = f(c_i)-f(c^*). \qedhere
	\end{align*}
\end{proof}

\begin{proof}[Proof of Lemma \ref{lem:expgradient}]
	The vector $\tilde{g}(c_i)$ is normalized by definition except if an index $j$ is chosen such that $c_i = p_j$, in which case $\norm{\tilde{g}(c_i)} = 0$. Thus $\Ex{\norm{\tilde{g}(c_i)}^2}= \Ex{\norm{\tilde{g}(c_i)}}\leq 1$ holds. Also we have
	\begin{align*}
		\Ex{\tilde{g}(c_i)} &= \sum_{j=1}^N \frac{1}{N} \cdot \frac{c_i-p_j}{\norm{c_i-p_j}} = \frac{1}{N} \sum_{j=1}^N \frac{c_i-p_j}{\norm{c_i-p_j}} = \frac{g(c_i)}{N},
	\end{align*}
	cf. Lemma \ref{lem:subgrad1}.
\end{proof}

\begin{proof}[Proof of Theorem \ref{thm:gradboundprob}]
	Assume that we have reached a center $c_i\in\mathbb{R}^d$ while running Algorithm~\ref{alg:subgradient}. Recall from Lemma \ref{lem:expgradient} that $\Exp{\tilde{g}}{\norm{\tilde{g}(c_i)}^2\mid c_i}\leq 1$ and $\Exp{\tilde{g}}{{\tilde{g}(c_i)}\mid c_i}=g(c_i)/N$. We have
	\begin{align*}
		%&\phantom{\;=\;}
		\Exp{\tilde{g}}{\norm{c_{i+1}-c^*}^2 \mid c_i}
		&= \Exp{\tilde{g}}{\norm{c_i- s \tilde{g}(c_i) - c^*}^2 \mid c_i} \\
		&= \Exp{\tilde{g}}{ \norm{c_i- c^*}^2  + \norm{s \tilde{g}(c_i)}^2 -2\left \langle s \tilde{g}(c_i) , c_i-c^* \right\rangle \mid c_i} \\%Using expected value over the choice of $\tilde{g}$, and the linearity of expectation, we further have \Exp{\tilde{g}}{\norm{c_{i+1}-c^*}^2 \mid c_i}
		&\leq \Exp{\tilde{g}}{\norm{c_i- c^*}^2 \mid c_i} + s^2 -2s\left \langle \Exp{\tilde{g}}{\tilde{g}(c_i) \mid c_i}, c_i-c^* \right\rangle \\
		&= \Exp{\tilde{g}}{\norm{c_i- c^*}^2 \mid c_i} + s^2 -\frac{2s}{N}\left \langle {g(c_i)}, c_i-c^* \right \rangle. 
	\end{align*}
	The law of total expectation implies, by taking expectations over $c_i$ on both sides and rearranging, that
	\begin{align*}
		\Exp{\tilde{g}}{ \norm{c_i- c^*}^2 } + s^2 \geq \Exp{\tilde{g}}{\norm{c_{i+1}-c^*}^2 } + \frac{2s}{N}  \Exp{c_i}{\left\langle g(c_i), c_i-c^* \right\rangle}.
	\end{align*}
	We sum the equations for $i\in \br{0,\ldots, \ell}$ such that the terms $\Exp{\tilde{g}}{\norm{c_{k}-c^*}^2 }$ for $k\in \br{1,\ldots,\ell}$ on both sides cancel and continue this derivation. We have
	\begin{align*}
		%&\phantom{\stackrel{\text{Jensen}}{\geq}} 
		\norm{c_0- c^*}^2   + (\ell+1) s^2 &\geq \Exp{\tilde{g}}{\norm{c_{\ell+1}-c^*}^2 } + \frac{2s}{N} \sum\limits_{i=0}^\ell \Exp{c_i}{\left\langle g(c_i), c_i-c^* \right\rangle} \\
		&\geq \frac{2s}{N}  \sum\limits_{i=0}^\ell \Exp{c_i}{\left\langle g(c_i), c_i-c^* \right\rangle}  \stackrel{\text{Subgr.}}{\geq} \frac{2s}{N}  \sum\limits_{i=0}^\ell \Exp{c_i}{ f\left(c_i\right)-f\left(c^*\right) } \\ 
		&\geq \frac{2s}{N}  \left( \ell+1\right) \cdot \min\limits_{i\in\lbrace 0,\ldots ,\ell\rbrace} \Exp{c_i}{ f\left(c_i\right)-f\left(c^*\right) } \\
		&\geq \frac{2s(\ell+1)}{N} \cdot \Exp{{c_i}}{ \min\limits_{i\in\lbrace 0,\ldots ,\ell\rbrace} f\left(c_i\right)-f\left(c^*\right) }
	\end{align*}
	using the subgradient property and the fact that for any set of positive real valued random variables $X_i$ we have $\forall i\colon \min_j(X_j)\leq X_i$. This implies $\forall i\colon\Ex{\min_j(X_j)}\leq \Ex{X_i}$ and thus $\Ex{\min_j(X_j)}\leq \min_j\Ex{X_j}$.
	Rearranging and substituting $R=\norm{c_0 - c^*}$ now implies 
	\begin{align*}
		&\Exp{\tilde{c}}{ f\left(\tilde c\right) -f\left(c^*\right) } = \Exp{{c_i}}{ \min\limits_{i\in \lbrace 0,\ldots,\ell\rbrace} f\left(c_i\right) -f\left(c^*\right) } \leq N\cdot \frac{R^2 + (\ell+1) s^2}{2 (\ell+1) s}. \qedhere
	\end{align*}
\end{proof}

\begin{proof}[Proof of Lemma \ref{lem:startingpoint}]
	Define the random variable $X=m(c^*,P)$. Clearly we have $X\geq 0$. Its expectation equals $\Ex{X}=\sum\nolimits_{i=1}^N \Pr{P=P_i}\cdot m(c^*,P_i)=\sum\nolimits_{i=1}^N m(c^*,P_i)/N=f(c^*)/N$. Thus, by Markov's inequality we have $\Pr{X > f(c^*)/(\delta_1 N)} \leq \delta_1$. Now choose an arbitrary $c_0 \in P$. We have $R=\norm{c_0-c^*}\leq m(c^*,P)\leq f(c^*)/(\delta_1 N)$ with probability at least $1-\delta_1$.
\end{proof}

\begin{proof}[Proof of Lemma \ref{lem:Restimate}]
	Let $\Delta=f(c^*)/N$. Let $c_0$ be the initial center chosen as described in Lemma \ref{lem:startingpoint} with absolute constant $\delta_1 = 1/8 \geq 1/81 > \eps^2$. Consider the two balls $B_1(c_0,\eps \Delta)$ and $B_2(c^*,\left(1/\eps^2\right) \Delta)$. Then $\norm{c_0-c^*} \leq 8\Delta$ holds with constant probability $1-\delta_1$ and clearly $c_0\in B_2$.
	
	Let $\mathcal{Q}$ consist of all sets of $\mathcal{P}$ that are fully contained in $B_2$. We have $|\mathcal{Q}|\geq (1-\eps^2)N$, since otherwise $f(c^*) \geq \sum_{P\in\mathcal{P}\setminus\mathcal{Q}} m(c^*,P) > \left(\eps^2 N\right) \cdot \left( \Delta/\eps^2\right) = f(c^*)$.
	
	Now, sample a collection $\mathcal{S}$ of $1/\eps$ sets, each uniformly from $\mathcal P$. All of our samples are completely contained in the ball $B_2$ with constant probability. By a union bound over the elements of $\mathcal S$, the probability that this fails is at most $\delta_2\leq \eps^2 \cdot 1/\eps = \eps < 1/8$.
	
	Now, let $\tilde R=\sum_{P\in \mathcal{S}} m(c_0,P)$ be our estimate. {Clearly we can compute it in time $\O{dn/\varepsilon}$.} We need to show that $\tilde R$ is close to the average cost $\Delta$ as we have claimed. To this end, consider the following two cases:
	\begin{description}
		\item[At most $(1-2\eps)|\mathcal{Q}|$ sets of $\mathcal{Q}$ are completely contained in $B_1$] In this case we have a constant probability $1-\delta_3$, for $\delta_3<2\eps<2/8$, that there is a set $Q \in \mathcal{Q}\cap \mathcal{S}$ that contains some point $q\in Q$ that lies outside $B_1$. Therefore we have $\tilde{R}\geq m(c_0,Q)\geq\norm{c_0-q}\geq \eps \Delta$. For the upper bound, note that the diameter of $B_2$ is $2\Delta/\eps^2$. Since all our samples are contained in that ball, we have $\tilde R \leq 2 \Delta|\mathcal{S}|/\eps^2=2\Delta/\eps^3$. Thus claim a) holds.
		
		\item[At least $(1-2\eps)|\mathcal{Q}|$ sets of $\mathcal{Q}$ are completely contained in $B_1$] Suppose the second item b) does not hold, i.e. we have $R=\norm{c_0-c^*} > 4 \eps \Delta$. We can bound the number of sets that are not fully contained in $B_1$ by $|\mathcal P \setminus \mathcal Q| + 2\eps |\mathcal{Q}|\leq \eps^2 N + 2\eps N \leq  3\eps N$. Let $\mathcal{B}=\lbrace P\in \mathcal{P}\mid P\subseteq B_1\rbrace$ be the remaining family of sets in $\mathcal P$ that are fully contained in $B_1$. Clearly $|{\mathcal{B}}|\geq (1-3\eps)N$. Now we compare the cost of using the center $c_0$ to the optimal cost. For every $P\in\mathcal P$ we have $|m(c^*,P)-m(c_0,P)|\leq \norm{c_0-c^*}$ by the triangle inequality. {For each $P\in \mathcal{B}$ it holds that}
		%the sets that are contained in $B_1$ we have 
		$m(c^*,P)-m(c_0,P)\geq (\norm{c_0-c^*}-\eps \Delta) - m(c_0,P) \geq \norm{c_0-c^*}-2\eps \Delta$. 
		So using $\norm{c_0-c^*} > 4 \eps \Delta$ and $\eps<1/9$ we can deduce 
		\begin{align*}
			f(c^*)-f(c_0) &= \sum\nolimits_{P\in {\mathcal{B}}} \left(m(c^*,P)-m(c_0,P)\right) + \sum\nolimits_{P\in \mathcal{P}\setminus {\mathcal{B}}} \left(m(c^*,P)-m(c_0,P)\right) \\
			&\geq |{\mathcal{B}}|(\norm{c_0-c^*}-2\eps \Delta) - (N-|{\mathcal{B}}|) \norm{c_0-c^*} \\ 
			&\geq 2|{\mathcal{B}}|\cdot \left(\norm{c_0-c^*}-\eps \Delta \right) - N\cdot \norm{c_0-c^*} \\ &\geq 
			2N \cdot \left(1-3\eps\right) \cdot \left(\norm{c_0-c^*}-\eps \Delta \right) - N\cdot \norm{c_0-c^*} \\
			&=
			N\cdot \left[ \norm{c_0-c^*}\cdot \left( 1-6\eps\right) - \eps \Delta \cdot \left( 2-6\eps\right) \right] \\ &\geq 
			N\cdot \left[ 4\eps \Delta\cdot \left( 1-6\eps\right) - \eps \Delta \cdot \left( 2-6\eps\right) \right] =
			2\eps \Delta N \left( 1-9\eps\right)>0,
		\end{align*} which contradicts the optimality of $c^*$. Thus, claim b) holds in this case. \qedhere
	\end{description}
\end{proof}

\begin{proof}[Proof of Theorem \ref{thm:thorup}]
	Let $b$ be an arbitrary center in $\mathcal{C}$ with
	\begin{equation}
	\label{eqn:indykpoint}
	\sum_{P\in \mathcal{P}} m(b,P) > (1+\eps) \sum_{P\in \mathcal{P}} m(\hat c,P).
	\end{equation}
	If there is no such center then all centers are good approximations, in which case the theorem is trivial. There are at most $|\mathcal C|$ choices for $b$. We study the random variable 
	\begin{equation}
	\label{eqn:indykvariable}
	X=\sum_{Q\in \mathcal{Q}} \frac{m(b,Q) - m(\hat c,Q)+m(\hat c,b)}{2m(\hat c,b)} = \sum_{Q\in \mathcal{Q}} h(Q),
	\end{equation}
	where $m(\hat c,b)=\norm{\hat c-b}$, and $h(Q)$ denote the summands of Equation~(\ref{eqn:indykvariable}). Since $m$ is a metric, by triangle inequality it holds that $X$ is the sum of random variables between $0$ and $1$. The \emph{bad} event is $X\leq |\mathcal{Q}|/2$. 
	
	If we denote by $\mathds{1}_{Q\in\mathcal{Q}}$ the indicator function that $Q\in\mathcal{Q}$, then we have $X= \sum_{Q\in\mathcal{Q}} h(Q) = \sum_{Q\in \mathcal{P}} h(Q)\cdot \mathds{1}_{Q\in\mathcal{Q}}$, and it holds that
	\begin{align*}
		\Ex{X} = \sum_{Q\in\mathcal{P}} \Ex{h(Q)\cdot  \mathds{1}_{Q\in\mathcal{Q}}} = \sum_{Q\in\mathcal{P}} \left( \mathds{1}_{Q\in\mathcal{Q}} \cdot \frac{|\mathcal{Q}|}{|\mathcal{P}|} \cdot h(Q)\right) = \frac{|\mathcal{Q}|}{|\mathcal{P}|} \sum_{Q\in\mathcal{P}} h(Q).
	\end{align*}
	
	Equation~(\ref{eqn:indykpoint}) and the triangle inequality for any set $Q\in\mathcal{P}$: $m(b,Q)+m(\hat c,Q)\geq m(\hat c, b)$, imply that
	\[
	\left( 2+\eps\right) \sum_{Q\in\mathcal{P}} m(b,Q) > \left( 1+\eps\right) \sum_{Q\in\mathcal{P}} m(\hat c,Q) + \left( 1+\eps\right) \sum_{Q\in\mathcal{P}} m(b,Q)  \geq \left( 1+\eps\right) \sum_{Q\in\mathcal{P}} m(\hat c, b).
	\]
	{Thus it holds that }
	\[
	\sum_{Q\in\mathcal{P}} \left( m(b,Q)-m(\hat c,Q)\right) > \frac{\eps}{1+\eps} \sum_{Q\in\mathcal{P}} m\left(b,Q\right) 
	{ > \frac{\eps}{2+\eps}\sum_{Q\in\mathcal{P}} m(\hat c,b)}.
	\]
	It follows that
	\[
	\Ex{X} = \frac{|\mathcal{Q}|}{|\mathcal{P}|} \sum_{Q\in\mathcal{P}} h(Q) = \frac{|\mathcal{Q}|}{|\mathcal{P}|} \sum_{Q\in\mathcal{P}} \left( \frac{m(b,Q)-m(\hat c,Q)}{2m(\hat c,b)} + \frac{1}{2}\right) 
	> \frac{|\mathcal{Q}|}{2{|\mathcal{P}|}} \sum_{Q\in\mathcal{P}} \left( \frac{\eps}{2+\eps} + 1\right), 
	\]
	and the \emph{bad} event is bounded by
	\begin{equation*}
		%\label{eqn:indykbadevent}
		\frac{|\mathcal{Q}|}{2} < \Ex{X}\cdot \frac{2+\eps}{2+2\eps} = \Ex{X}\cdot \left( 1-\frac{\eps}{2+2\eps}\right) = \Ex{X}\cdot \left( 1-\gamma \right),
	\end{equation*}
	where $\gamma = \eps/\left( 2+{2}\eps\right)$, and $\gamma\geq \eps/4$ holds since $\eps<1$. Using a Chernoff bound \cite{MotwaniR95} we have that 
	\begin{equation}
	\label{eqn:chernoff}
	\Pr{X \leq \frac{|\mathcal{Q}|}{2}} \leq \Pr{X < (1-\gamma)\cdot \Ex{X}} < e^{-\gamma^2 \Ex{X}/2} \leq e^{-\eps^2 |\mathcal{Q}|/64},
	\end{equation}
	since $\gamma \geq \eps/4$ and $\gamma^2/(1-\gamma) \geq \eps^2 /16$. Transforming $X$ in Equation~(\ref{eqn:chernoff}) we have
	\[
	\Pr{\sum_{Q\in\mathcal{Q}} m(b,Q) \leq \sum_{Q\in\mathcal{Q}} m(\hat c,Q) } < 
	e^{ -\eps^2 |\mathcal{Q}| /64}.
	\]
	Defining the set of \emph{bad centers} $\mathcal B = \lbrace b\in \mathcal{C}\mid \sum_{P\in\mathcal{P}} m(b,P) > (1+\eps)\sum_{P\in\mathcal{P}} m(\hat c,P)\rbrace$, we finally have that 
	\begin{align*}
		&\Pr{\forall b\in \mathcal{C}, \sum_{P\in\mathcal{P}} m(b,P) > (1+\eps)\sum_{P\in\mathcal{P}} m(\hat c,P) \colon \sum_{Q\in\mathcal{Q}} m(b,Q)> \sum_{Q\in \mathcal{Q}} m(\hat{c},Q)} \\
		&= 1- \Pr{\exists b\in \mathcal{C}, \sum_{P\in\mathcal{P}} m(b,P) > (1+\eps)\sum_{P\in\mathcal{P}} m(\hat c,P) \colon \sum_{Q\in\mathcal{Q}} m(b,Q) \leq \sum_{Q\in \mathcal{Q}} m(\hat{c},Q)} \\
		&\geq 1-\sum_{b\in \mathcal B} \Pr{\sum_{Q\in\mathcal{Q}} m(b,Q) \leq \sum_{Q\in \mathcal{Q}} m(\hat{c},Q)} \geq 1-|\mathcal B| \cdot e^{-\varepsilon^2 |\mathcal{Q}|/64} \geq 1-|\mathcal{C}| \cdot e^{-\varepsilon^2 |\mathcal{Q}|/64}.
	\end{align*}
	Then it holds in particular for $a\in \mathcal{C}$ that with probability at least $1-|\mathcal{C}|\exp(-\varepsilon^2 |\mathcal{Q}|/64)$
	\[
	\left( \sum_{P\in\mathcal{P}} m(a,P) > (1+\eps)\sum_{P\in\mathcal{P}} m(\hat c,P)\right)
	\Rightarrow 
	\left( \sum_{Q\in\mathcal{Q}} m(a,Q)> \sum_{Q\in \mathcal{Q}} m(\hat{c},Q)\right).
	\]
	But it holds that $\sum_{Q\in\mathcal{Q}} m(a,Q) \leq  \sum_{Q\in \mathcal{Q}} m(\hat{c},Q)$ by optimality of $a$, so the contrapositive yields that 
	\[
	\Pr{\sum_{P\in\mathcal{P}} m(a,P) \leq (1+\eps)\sum_{P\in\mathcal{P}} m(\hat c,P)} \geq 1-|\mathcal{C}|\cdot e^{-\varepsilon^2 |\mathcal{Q}|/64}. \qedhere
	\]
\end{proof}

\begin{proof}[Proof of Theorem \ref{thm:mainresult}]
	Set $\ell=(68/\eps)^2$. Using Lemma \ref{lem:Restimate}, either our initial center $c_0$ satisfies $\norm{c_0-c^*}\leq 4 \eps f(c^*)/N$. In that case Lemma~\ref{lem:nlipschitz} yields $f(c_0) \leq f(c^*)+N\norm{c_0-c^*}\leq (1+4\eps)f(c^*)$, so the starting point is already a good center.
	%the triangle inequality. 
	
	Otherwise we have $\eps f(c^*)/N \leq \tilde{R} \leq \left( 2/\eps^3\right)\cdot f(c^*)/N$ and $R=\norm{c_0-c^*}\leq 8 f(c^*)/N$ by Lemma~\ref{lem:Restimate}. Thus $\eps^3 \tilde{R}/2 \leq f(c^*)/N \leq \tilde{R}/\eps$. To improve this we run the main loop of Algorithm \ref{def:generalmedian} for the step sizes $s=\tilde{R}_j/\sqrt{\ell+1}$, where $\tilde R_j=2^{j-1}\cdot \eps^3 \tilde{R}$ for all values of $0\leq j \leq \lceil\log(2/\eps^4)\rceil$.  For some particular value of $j$ we have a $2$-approximation given by $f(c^*)/N \leq \tilde{R}_j \leq 2 f(c^*)/N$. In this particular run, setting the step size $s=\tilde{R}_j/\sqrt{\ell+1}$ and plugging this into the bound given in Theorem~\ref{thm:gradboundprob} we have that
	\begin{align*}
		%\Exp{\tilde{c}}{f(\tilde{c}) - f(c^*)}
		%&\leq \frac{f(c^*)}{\sqrt{\ell+1}}\cdot \left( 2+\frac{8^2}{2}\right) =\frac{34 f(c^*)}{\sqrt{\ell+1}}.\\
		\Exp{\tilde{c}}{ f(\tilde{c}) -f\left(c^*\right) } 
		&\leq N \cdot \frac{R^2 + (\ell+1) s^2}{2(\ell+1)s} \leq N \cdot \frac{R^2 + \tilde{R}_j^2}{2\sqrt{\ell+1}\tilde{R}_j} \leq \frac{8^2 + 2^2}{2\sqrt{\ell+1}} f(c^*) \leq \frac{\eps}{2} f(c^*).
	\end{align*}
	Using Markov's inequality we have that
	\begin{align*}
		\Pr{f(\tilde{c})-f(c^*)\geq 4\eps f(c^*)} \leq \frac{\eps f(c^*)}{8\eps f(c^*)} = \frac{1}{8} = \delta_4.
	\end{align*}
	The best center collected in all repetitions cannot be worse than this particular $\tilde c$ or $c_0$, see the cases of Lemma \ref{lem:Restimate}.
	
	Finally we have a collection $\mathcal{C}$ of $|\mathcal{C}|\in O(1/\eps^2 \cdot \log 1/\eps)$ centers and want to find one of them that is a $(1+\eps)$ approximation for the best center in $\mathcal{C}$ via Theorem~\ref{thm:thorup}. We sample a collection of $64/\eps^2 \cdot \ln(8|\mathcal{C}|)\in O(1/\eps^2 \cdot \log 1/\eps)$ point sets from $\mathcal P$ and find the best center for this subset of points which is the final output of our algorithm. By Theorem~\ref{thm:thorup} this center is within another factor of $(1+\eps)$ to the best in $\mathcal{C}$ with failure probability at most $\delta_5\leq 1/8$. The total approximation factor is thus at most $(1+4\eps)(1+\eps)\leq 1+9\eps$. Rescaling $\eps$ yields the correctness. The total failure probability is at most $\delta=\sum_{i=1}^5 \delta_i \leq 6/8$ by a union bound over all bad events in Lemma~\ref{lem:Restimate} and this theorem.
	
	We continue with the running time. The initial center $c_0$ and the estimate $\tilde R$ can be computed in $O(d{n}/\eps)$ time, see Lemma~\ref{lem:Restimate}. The main loop of Algorithm \ref{alg:subgradient} takes $O(d{n})$ in each iteration and runs for $\ell\in O(1/\eps^2)$ iterations for a fixed step size. But we try $O(\log 1/\eps)$ different step sizes. This makes up a running time of $O(d{n}/\eps^2 \cdot\log 1/\eps)$. Finally we evaluate the objective function for $O(1/\eps^2\log 1/\eps)$ centers for the sample of $O(1/\eps^2 \cdot \log 1/\eps)$ sets taken via Theorem~\ref{thm:thorup}. This can be done in time $O(d{n}/\eps^4 \cdot \log^2 1/\eps)$ which dominates the running time as we have claimed.
\end{proof}

\begin{proof}[Proof of Theorem \ref{lem:datastructurelower}]
	We reduce from the indexing problem which is known to have $\Omega\left( n\right)$ one-way randomized communication complexity \cite{KremerNR99}. Alice is given a vector $a\in \lbrace 0,1\rbrace^n$. Bob has an index $i\in [n]$, and has to guess the value of the $i$-th bit of $a$, denoted $a_i$, with probability $2/3$.
	
	It is known from \cite{AgarwalS15} that there is a centrally symmetric point set $K$ of size $\Omega\left( \exp\left(d^{1/3}\right)\right)$ on the unit hypersphere in $\REAL^d$ centered at the origin, such that for any pair of distinct points $p,q\in K$ it holds that
	\[
	\sqrt{2}\left( 1-2/d^{1/3}\right) \leq \norm{p-q} \leq \sqrt{2}\left( 1+2/d^{1/3}\right)
	\]
	unless $p\neq -q$, in which case $\norm{p-q}=2$.
	
	Let $d$ be the smallest integer such that $d\geq 8$ and $n \leq \exp\left( d^{1/3}\right)$. We choose a set of $\exp\left(d^{1/3}\right)$ pairs of centrally symmetric points of $K$. We may assume that there is a lexicographic order of these pairs, so there is a mapping between the indices of $a$ and the pairs of points of $K$ that is known to both, Alice and Bob. Alice constructs the set $S$ by including the first point of the $i$-th pair, denoted $p_i$, if and only if $a_i=1$. She builds a data structure $\Sigma_S$ which she sends to Bob.
	
	Let the data structure be such that for any $x\in \REAL^d$ the answer to a query $\Sigma_S (x)$ satisfies
	\begin{align*}
		\frac{m(x,S)}{\alpha} \leq \Sigma_S (x) \leq m(x,S),
	\end{align*}
	for some constant $1<\alpha<\sqrt{2} \left( 1-2/d^{1/3}\right)$.
	We consider two cases:
	\begin{itemize}
		\item If $a_i=1$ then $p_i$ is included in $S$. Thus $\Sigma_S \left(-p_i\right) \geq m\left(-p_i, S\right) / \alpha = 2/\alpha$. Since $\alpha < \sqrt{2} \left( 1-2/d^{1/3}\right)$ it holds that $\Sigma_S \left(-p_i\right) >  \sqrt{2} \left( 1+2/d^{1/3}\right)$.
		\item If $a_i=0 $ then $p_i \notin S$ and thus $\Sigma_S \left(-p_i\right) \leq m\left( -p_i, S\right) \leq \sqrt{2}\left( 1+2/d^{1/3}\right).$
	\end{itemize}
	
	\noindent
	Thus if $\alpha < \sqrt{2} \left( 1-2/d^{1/3}\right)$ Bob could based on $\Sigma_S$ solve the indexing problem by querying $q_i = -p_i$. Consequently any encoding of $\Sigma_S$ uses $\Omega\left(\min\lbrace n, \exp \left(d^{1/3}\right)\rbrace \right)$ bits of space.
\end{proof}

\begin{proof}[Proof of Theorem \ref{thm:SEBalgorithm}]
	The correctness of the algorithm follows from \cite{MunteanuSF14} and Theorem \ref{thm:mainresult}. It remains to analyze the running time.
	In the first case we go through all input distributions and use $k$ independent copies of a weighted reservoir sampler \cite{Chao82,Efraimidis15} to get the $k$ samples. This takes $O(dnzs)\subseteq O(dnz/\eps^2 \cdot \log 1/\eps)$ time. The subsampled problem is then solved via Theorem~\ref{thm:mainresult} with ${n}=1$ in time $O(d/\eps^4 \cdot \log^2 1/\eps)$ with failure probability at most $\delta=\sum_{i=1}^5 \delta_i \leq 6/8$.
	
	In the second case, each realization can be sampled similarly in time $O(dnz)$ but the probability that a realization is non-empty can only be lower bounded by $\eps$. However this means that the expected number of samples that we need to take in order to have $k$ non-empty realizations is bounded by at most $k/\eps$. Thus, by an application of Markov's inequality the probability that we need more than $8k/\eps\in O(1/\eps^3 \cdot \log 1/\eps)$ trials is bounded by at most $\delta_6 = 1/8$. So we can assume with constant probability that this step succeeds in time $O(dnz/\eps^3 \cdot \log 1/\eps)$. The subsampled problem is then solved via Theorem~\ref{thm:mainresult} 
	in time $O(dn/\eps^4 \cdot \log^2 1/\eps)$. The final failure probability is at most $\delta=\sum_{i=1}^6 \delta_i \leq 7/8$.\qedhere 
\end{proof}

\begin{proof}[Proof of Theorem \ref{thm:algsvdd}]
	With the described adaptations the correctness of the algorithm follows from Theorem~\ref{thm:SEBalgorithm}.
	
	The running time increases by a factor that is imposed by the simulation of the distance computations within Algorithm~\ref{alg:subgradient}. Note that by the invariant there are at most $i+1$ non-zero coefficients in the $i$-th step. Expression (\ref{eqn:dist_implicit}) can thus be evaluated in time $O(i^2 d)$ assuming $K$ can be evaluated in time $O(d)$. We conclude:
	\begin{itemize}
		\item The sampling part of Algorithm~\ref{alg:SEB} does not change and thus runs in time $\O{dnz/\eps^3 \log 1/\eps}$.
		
		\item Estimating $\tilde R=\sum_{P\in \mathcal{S}} m(c_0,P)$ takes time $O(dn/\eps)$ since $i=0$, $|\mathcal S|=O(1/\eps)$, and for each $P$ we have $|P|\leq n$.
		
		\item The subgradient computation in the $i$-th iteration of the main loop takes $O(d n i^2)$ time since it needs to maximize over $n$ distances. This means that for $\ell\in \O{1/\eps^2}$ iterations we need $\O{dn\sum_{i=1}^\ell i^2} = \O{dn\ell^3} = \O{dn/\eps^6}$ time. This is repeated $\O{\log 1/\eps}$ times, which implies a running time of $\O{dn/\eps^6 \cdot \log 1/\eps}$.
		
		\item The evaluation of the minimum at the end reaches as before over $\O{1/\eps^2 \cdot \log 1/\eps}$ centers, each defined by $O(\ell)$ non-zero coefficients. Each is evaluated with respect to a sample of $\O{1/\eps^2 \cdot \log 1/\eps}$ sets from the input. Since each maximum distance evaluation takes $\O{dn \ell^2} = \O{dn/\eps^4}$ time, we have that the total time for evaluation is $\O{dn/\eps^8 \cdot \log^2 1/\eps}$.
	\end{itemize}
	Summing the running times yields the claim.
\end{proof}

\end{document}